\documentclass[review]{elsarticle}

\usepackage{amsmath,amsthm,amssymb,latexsym,stmaryrd}
\usepackage{tikz}
\usepackage{tikz-cd}
\usepackage[mathscr]{euscript} 
\usepackage{stackrel}
\usepackage[all]{xy}

\journal{Journal of \LaTeX\ Templates}

\begin{document}

\begin{frontmatter}

\title{Solving Homotopy Domain Equations \tnoteref{mytitlenote}}


\author{Daniel O. Martínez-Rivillas}

\author{Ruy J.G.B. de Queiroz}


\begin{abstract}
In order to get $\lambda$-models with a rich structure of $\infty$-groupoid, which we call ``homotopy $\lambda$-models'', a general technique is described for solving domain equations on any Cartesian closed $\infty$-category (c.c.i.) with enough points. Finally, the technique is applied in a particular c.c.i., where some examples of homotopy $\lambda$-models are given. 
\end{abstract}


\end{frontmatter}


\newtheorem{defin}{Definition}[section]
\newtheorem{teor}{Theorem}[section]
\newtheorem{corol}{Corollary}[section]
\newtheorem{prop}{Proposition}[section]
\newtheorem{rem}{Remark}[section]
\newtheorem{lem}{Lemma}[section]
\newtheorem{n}{Notation}[section]
\newtheorem{ejem}{Example}[section]

\section{Introduction}
\label{chap:capitulo1}

The purpose of this work is to give a follow-up on the project of generalisation of Dana Scott's Domain Theory \cite{DBLP:books/mk/Abramsky94} and \cite{DBLP:books/mk/Asperti91}, to a Homotopy Domain Theory, which began in \cite{MartinezHoDT}, in the sense of providing methods to find $\lambda$-models that allow raising the interpretation of equality of $\lambda$-terms (e.g., $\beta$-equality, $\eta$-equality, etc.) to a semantics of higher equalities.

\bigskip Dana Scott's Domain Theory \cite{DBLP:books/mk/Abramsky94}, was created in the 1960s and 1970s with the purpose of finding a semantics for programming languages, specifically for functional programming languages theorised by type-free $\lambda$-calculus. Where each computational program, which corresponds to a $\lambda$-term, is interpreted as a continuous function $f:D\rightarrow D$ defined on an ordered set (poset) $D$ called Domain, which would come a solution of the Recursive Domain Equation 
$$X\cong [X\rightarrow X],$$ 
hence $D$ is a model of type-free $\lambda$-calculus or also called $\lambda$-model, where $D_\infty$ is the first non-trivial $\lambda$-model introduced by Dana Scott. Thus, a computable construction can be more easily verified in its interpretation as a continuous function.

\bigskip In addition to the many recursively defined computational objects, type-free $\lambda$-calculus is a classic example of a programming language with recursive definitions of  data-types (a classification that specifies the type of value a variable),  since the expression $D\cong [D\rightarrow D],$ implies a recursive definition of the data-types, in the sense that $D$ is the limit of a non-decreasing sequence of posets $D_i\subseteq D$ recursively defined, where each $D_i$ interprets a type. Thus, a recursively defined computational object (e.g., the factorial function) can be seen as a recursive sequence of partial functions $f_i:D\rightarrow D$ which converges to a total function $f:D\rightarrow D$ being the interpretation of the $\lambda$-term which represents the computational object. 

\bigskip In the 1970s, $\lambda$-calculus was extended by Martin Löf to Intuitionistic Type Theory (ITT) with the purpose of formalising mathematical proofs to computer programmes, i.e., it emerged as a computational alternative to the fundamentals of mathematics. Later, in the 2010s, the Homotopy Type Theory (HoTT) was proposed as an interaction between ITT and Homotopy Theory, where we can affirm that HoTT is the most recent version of the theorization of the notion of an algorithm (in the sense of constructive mathematics seen as algorithms). Finally in \cite{Queiroz2016}, \cite{Ramos2017} and \cite{Thiago2023}, an alternative to HoTT based on computational paths; as a finite sequence of rewriting of computational terms according to the ITT axioms.

\bigskip Thus, as there is a semantics of type-free $\lambda$-calculus with Dana Scott's Domain Theory, we here want to propose a semantics of a untyped version of HoTT based on computational paths starting from a ``Domain Theory of Homotopy'', that is, we want to start the way for an interpretation of the most recent notion of algorithm given by HoTT, starting with the untyped case.

\bigskip The initiative to search for $\lambda$-models \cite{DBLP:books/mk/HindleyS08} with a structure of $\infty$-groupoid emerged in \cite{Martinez}, which studied the geometry of any complete partial order (c.p.o) (e.g., $D_\infty$), and found that the topology inherent in these models generated trivial higher-order groups. From that moment on, the need arose to look for a type of model that presented a rich geometric structure; where their higher-order fundamental groups would not collapse. 

\bigskip It is known in the literature that Dana Scott's Domain-Theory \cite{DBLP:books/mk/Abramsky94} and \cite{DBLP:books/mk/Asperti91}, provides general techniques for obtaining $\lambda$-models by solving domain equations over arbitrary  Cartesian closed categories. To fulfil the purpose of getting $\lambda$-models with non-trivial $\infty$-groupoid structure, the most natural way would be to adapt Dana Scott's Domain Theory to a ``Homotopy Domain Theory''. Where the Cartesian closed categories (c.c.c) will be replaced by Cartesian closed $\infty$-categories (c.c.i), the c.p.o's will be replaced by ``c.h.p.o's (complete homotopy partial orders)'',  the $0$-categories by $0$-$\infty$-categories and the isomorphism between objects in a Cartesian closed 0-categories (at the recursive domain equation) will be replaced by an equivalence between objects in a Cartesian closed $0$-$\infty$-category, which we call ``homotopy domain equation'' such as summarized in Table \ref{Table HoDT}. See \cite{Martinez21}.

	\begin{table}[ht]
		\caption{Generalization of Recursive Domain Equations} 
		\label{Table HoDT}
		\centering
		\begin{tabular}{p{6cm}p{5.5cm}}
			\hline
			\multicolumn{1}{c}{\textbf{Dana Scott's Domain-Theory}} & \multicolumn{1}{c}{\textbf{Homotopy Domain Theory}}  \\
			\hline     
			Cartesian Closed Category (c.c.c)  & Cartesian Closed $\infty$-Category (c.c.i) 
			\\ 
			
			Complete Partial Order (c.p.o) &  Complete Homotopy Partial Order (c.h.p.o)           
			\\ 
			0-category & $0$-$\infty$-category \\
			Recursive Domain Equation: 
			
			$X\cong (X\Rightarrow X)$	& Homotopy Domain Equation: 
			
			$X\simeq(X\Rightarrow X)$
			\\ \hline
		\end{tabular}
	\end{table}

Just as the recursive Domain Equation $X\cong [X\rightarrow X]$ (in the category of the c.p.o's) has an implicit recursive definition of the data-types, the ``Homotopy Domain Equation'' $X\simeq [X\rightarrow X]$ (in the ``$\infty$-category of the c.h.p.o's'') would also have a recursive definition of the data-types. A recursively defined computational object (e.g., a proof by mathematical induction) would being of a higher order relative to the classical case, whose interpretation would be recursively defined by a sequence of partial functors $F_i:K\rightarrow K$, over a  Kan complex $K$ weakly ordered, which converges to a total functor $F:K\rightarrow K$, whose details are not among the objectives of this paper, but will be developed in future works, when studying the semantics (case of inductive types) of the version of HoTT based on computational paths. \cite{Martinez2HoDTvsHoTT21}.

\bigskip Thus, intuitively, from the computational point of view, we have that a Kan complex, which satisfies the Homotopy Domain Equation, is not only capable of verifying the computability of constructions typical of classical programming languages, as $D_\infty$ does it, but also, it has the advantage (over $D_\infty$) of verifying the computability of higher constructions, such as a mathematical proof of some proposition, the proof of the equivalence between two proofs of the same proposition etc.

\bigskip To guarantee the existence of Kan complexes with good information, we define the ``non-trivial Kan complexes", which in intuitive terms, are those that have holes in all higher dimensions and also have holes in the locality or class of any vertex. We establish conditions to prove the existence of non-trivial Kan complexes which are reflexive, which we call ``non-trivial homotopy $\lambda$-models"  than the homotopic $\lambda$-models initially defined in \cite{Martinez}, and developed in \cite{MartinezHoDT}.  

\bigskip It should be clarified that the literature around Kan complexes is related to some computational theories, such as Homotopy Type Theory (HoTT) \cite{DBLP:books/mk/Univalent13}, so that to ensure the consistency of HoTT, Voevodsky \cite{Voevodsky} (see \cite{Lumsdaine} for higher inductive types) proved that HoTT has a model in the category of Kan complexes (see \cite{DBLP:books/mk/Univalent13}). 

\bigskip To meet the goal of getting homotopy $\lambda$-models, in Section \ref{Section: Complete Homotopy Partial Orders} we generalization the c.p.o to c.h.p.o (complete homotopy partial order), in Section \ref{HDE arbitrary cci}, we propose a method for solving homotopy domain equations on any c.c.i., namely: first, one solves the contravariant functor problem in a similar way to classical Domain Theory, and later one uses a version from fixed point theorem to find some solutions of this equation. Finally, in Section \ref{HDE on Kl(P)}, the methods of the previous section are applied in a particular c.c.i., and thus one ends up guaranteeing the existence of homotopy $\lambda$-models which, as previously mentioned, present an $\infty$-groupoid structure with relevant information in higher dimensions.

\section{Theorical Foundations}
\label{Section: Complete Homotopy Partial Orders}

In this section we reintroduce some concepts on $\infty$-category, Kan complex and complete homotopy partial order (c.h.p.o) as a direct generalization of the c.p.o's, where the sets are replaced  by Kan complexes and the order relations $\leq$ by weak order relations  $\precsim$ (\cite{MartinezHoDT} and \cite{Martinez2HoDTvsHoTT21}). 

\subsection{Definition of $\infty$-category and Kan complex}
\label{sub:Definiton-infty-categoy-Kan-complex}

A simplicial set $K$ is defined as a presheaf $\Delta^{op}\rightarrow Set$, with $\Delta$ being the \textit{simplicial indexing category}, whose objects are finite ordinals $[n]=\{0,1,\ldots,n\}$, and morphisms are the (non strictly) order preserving maps. $\Delta^n$ is the \textit{standard $n$-simplex} defined for each $n\geq 0$ as the simplicial set $\Delta^n:=\Delta(-,[n])$. And $\Lambda_i^n$ is a \textit{horn} defined as largest subobject of $\Delta^n$ that
does not include the face opposing the $i$-th vertex.

\begin{defin}[$\infty$-category \cite{DBLP:books/mk/Lurie}]
	An $\infty$-category is a simplicial set $X$ which has the following property: for any $0<i<n$, any map $f_0:\Lambda_i^n\rightarrow X$ admits an
	extension $f:\Delta^n\rightarrow X$.	
\end{defin}

\begin{defin}
	From the definition above, we have the following special cases:
	\begin{itemize}
		\item $X$ is a Kan complex if there is an extension for each $0\leq i\leq n$.
		\item $X$ is a category if the extension exists uniquely \cite{DBLP:books/mk/Rezk22}.
		\item $X$ is a groupoid if the extension exists for all $0\leq i\leq n$ and is unique \cite{DBLP:books/mk/Rezk22}. 
	\end{itemize}
\end{defin}

In other words,  a Kan complex is an $\infty$-groupoid; composed of objects, 1-morphisms, 2-morphisms, ...,  all those invertible.

\begin{defin}
A functor of $\infty$-categories $X\rightarrow Y$ is exactly a morphism of simplicial sets. Thus, $Fun(X,Y)=Map(X,Y)$ must be a simplicial set of the functors from $X$ to $Y$. 
\end{defin}

\begin{n}
The notation $Map(X,Y)$  is usually used for simplicial sets, while $Fun(X,Y)$ is for $\infty$-categories. One will refer to morphisms in $Fun(X, Y)$ as natural transformations of functors,
and equivalences in $Fun(X, Y)$ as natural equivalences. 

\bigskip The composition from n-simplex $f:\Delta^n\rightarrow X$ (or $f\in X_n$) with a functor $F:X\rightarrow Y$, will be denoted as the image $F(f)\in Y_n$, where $n\geq 0$.

\bigskip A 1-simplex $f:\Delta^1\rightarrow X$, such that $f(0)=x$ and $f(1)=y$ will be denoted as a morphism $f:x\rightarrow y$ in the $\infty$-category $X$.

\bigskip An inner horn $\Lambda_1^2\rightarrow X$, which corresponds to composable morphisms $x\xrightarrow{f}y\xrightarrow{g}z$
in the $\infty$-category $X$, will be denoted by $(g,-,f)$ or in some cases to simplify notation it will be denoted by $g.f$.
\end{n}

\begin{prop}[\cite{DBLP:books/mk/Lurie} and \cite{DBLP:books/mk/Cisinski}]
	For every $\infty$-category $Y$, the simplicial set $Fun(X, Y)$ is an $\infty$-category.
\end{prop}

\bigskip The proof of the following theorem can be found in \cite{DBLP:books/mk/Lurie} and \cite{DBLP:books/mk/Cisinski}.

\begin{teor}[Joyal]
	A simplicial set $X$ is an $\infty$-category if and only if the canonical morphism 
	$$Fun(\Delta^2,X)\rightarrow Fun(\Lambda_1^2,X),$$
	is a trivial fibration. Thus, each fibre of this morphism is contractible.
\end{teor}

The above theorem guarantees the laws of coherence of the composition of 1-simplexes or morphisms of an $\infty$-category. In the sense that the composition of morphisms is unique up to homotopy, i.e., the composition is well-defined up to a space of choices is contractible (equivalent to $\Delta^0$).

\begin{defin}[Space of morphisms \cite{DBLP:books/mk/Rezk17}]
For two vertices $x,y$ in an $\infty$-category $X$, define the space of morphisms $X(x,y)$ by the following pullback diagram	
	\[\xymatrix{
	& {X(x,y)}\ar[d]_{}\ar[r]^{} &
	Fun(\Delta^1,X) \ar[d]^{(s,t)}\\
	& \Delta^0\ar[r]_{(x,y)} & X\times X
	& 
}\]
in the category $sSet$.
\end{defin}

\begin{prop}[\cite{DBLP:books/mk/Lurie} and \cite{DBLP:books/mk/Rezk17}]
	The morphism spaces $X(x,y)$ are Kan complexes.
\end{prop}

\subsection{Natural transformations and natural equivalence}

In category theory, we have the concept of isomorphism of objects. For the case of the $\infty$-categories, we will have the equivalence of objects (vertices) in the following sense.  

\begin{defin}[Equivalent vertices]
A morphism (1-simplex) $f:x\rightarrow y$ in an $\infty$-category $X$ is invertible (an equivalence) if there is morphism $g:y\rightarrow x$ in $X$, a pair of 2-simplexes $\alpha,\beta\in X_2$ such that 
$(g,-,f)\xrightarrow{\alpha}1_x$ and $(f,-,g)\xrightarrow{\beta}1_y$, i.e., we can fill out following diagram
 	\[\xymatrix{
 	& {x}\ar[dd]_{1_x}\ar[rr]^f& &
 	y \ar[dd]^{1_y}\ar@{-->}[lldd]^g\\ \\
 	& x\ar[rr]_{f}& & y
 	& 
 }\]
with the 2-simplexes $\alpha$ and $\beta$.	
\end{defin}

\begin{teor}[\cite{DBLP:books/mk/Lurie} and \cite{DBLP:books/mk/Cisinski}]
	Let $X$ be an $\infty$-category. The following are equivalent
	\begin{enumerate}
		\item Every morphism (1-simplex) in $X$ is an equivalence.
		\item  $X$ is a Kan complex.
	\end{enumerate}
\end{teor}

\subsection{Natural transformations and natural equivalence}

\begin{defin}[\cite{DBLP:books/mk/Cisinski} and \cite{DBLP:books/mk/Rezk17}]
	If $X$ and $Y$ are $\infty$-categories, and if $F, G : X\rightarrow Y$ are two functors, a natural
	transformation from $F$ to $G$ is a map $H : X\times \Delta^1\rightarrow Y$ such that
	$$H(x,0)=F(x),\hspace{0.5cm}H(x,1)=G(x),$$
	for each vertex $x\in X$. Such a natural transformation is invertible or it is a natural equivalence if for any vertex $x\in X$, the
	induced morphism $F(x)\rightarrow G(x)$ (corresponding to the restriction of $H$ to
	$\Delta^1 \cong \{x\}\times\Delta^1$) is invertible in $Y$. If there is a natural equivalence from $F$ to $G$, we write $F\simeq G$.
\end{defin}

\begin{rem}
This means that for each vertex $x\in X$, one chooses a morphism $H_x:F(x)\rightarrow G(x)$ such that the following diagram 
 	\[\xymatrix{
	& {F(x)}\ar[dd]_{F(f)}\ar[rr]^{H_x}\ar@{-->}[rrdd]^g& &
	G(x) \ar[dd]^{G(f)}\\ \\
	& F(x')\ar[rr]_{H_{x'}}& & G(x')
	& 
}\]
commutes under some 2-simplexes $\alpha:g\rightarrow(G(f),-,H_x)$ and $\beta:(H_{x'},-,F(f))\rightarrow g$.
\end{rem}

\subsection{Categorical equivalences and homotopy equivalences}

\begin{defin}[Categorical equivalence \cite{DBLP:books/mk/Rezk17} and \cite{DBLP:books/mk/Lurie}]
A functor of $\infty$-categories $F:X\rightarrow Y$ is a categorical equivalence if there is another functor $G:Y\rightarrow X$, such that $GF\simeq 1_X$ and $FG\simeq 1_Y$.	
\end{defin}

\begin{rem}
	From the definition above, if  $F:X\rightarrow Y$ is a functor of Kan complexes, we say that $F$ is a homotopy equivalence.
\end{rem}

\begin{lem}[\cite{DBLP:books/mk/Rezk17} and \cite{DBLP:books/mk/Cisinski}]
	A functor of $\infty$-categories $F:X\rightarrow Y$ is a categorical equivalence
	if it satisfies the following two conditions:
	\begin{itemize}
		\item Fully Faithful (Embedding): For two objects $x,y\in X$ the induced functor of Kan complexes
		$$X(x,y)\rightarrow Y(Fx,Fy),$$
		is a homotopy equivalence.
		\item Essentially Surjective: For every object $y\in Y$ there exists an object $x\in Y$ such that $Fx$ is equivalent to $y$.
	\end{itemize}
\end{lem}

\subsection{The join of $\infty$-categories}
\label{sub:section3}

Next, the extension from join of categories to $\infty$-categories. This will make us able to define limit and colimit in an $\infty$-category.

\begin{defin}[Join \cite{DBLP:books/mk/Lurie}]
	Let $K$ and $L$ be simplicial sets. The join $K\star L$ is the simplicial set defined by
	$$(K\star L)_n:=K_n\cup L_n\cup\bigcup_{i+1+j=n}K_i\times L_j, \hspace{0.5cm}n\geq 0.$$ 
\end{defin}

\begin{ejem}\cite{DBLP:books/mk/Groth15}
	\begin{enumerate}
		\item If $K\in sSet$ and $L=\Delta^0$, then $K^\vartriangleright=K\star\Delta^0$ is the cocone or the right cone on $K$. Dually,  If $L\in sSet$ then $L^\vartriangleleft=\Delta^0\star L$ is the cone or the left cone on $L$.
		\item Let $K=\Lambda_0^2$. If we see this left horn as a pushout, the cocone $(\Lambda_0^2)^\vartriangleright$ is isomorphic to the square $\square=\Delta^1\times\Delta^1$, that is, to the filled in diagram
		 	\[\xymatrix{
			& {(0,0)}\ar[dd]_{}\ar[rr]^{}\ar[rrdd]& &
			(1,0) \ar[dd]^{}\\ \\
			& (0,1)\ar[rr]_{}& & (1,1)
			& 
		}\] 
	\end{enumerate}
\end{ejem}

\begin{prop}\cite{DBLP:books/mk/Lurie}
	\begin{enumerate}
		\item[(i)] For the standard simplexes one has an isomorphism $\Delta^i\star\Delta^j\cong \Delta^{i+1+j}$, $i,j\geq 0$, and these isomorphisms are with the obvious inclusions of $\Delta^i$ and $\Delta^j$.
		\item[(ii)] If $X$ and $Y$ are $\infty$-categories, then the join $X\star Y$ is an $\infty$-category.
	\end{enumerate}
\end{prop}

\subsection{The slice $\infty$-category}

In the case of the classical categories, if $A,B$ are categories and $p:A\rightarrow B$ is any functor, one can form the slice category $B_{/p}$ of the object over $p$ or cones on $p$. The following propositions allow us to define the slice $\infty$-category.   

\begin{prop}[\cite{Joyal2002}]
Let $K$ and $S$ be simplicial sets, and $p:K\rightarrow S$ be an arbitrary map. There is a simplicial set $S_{/p}$ such that there exists a natural bijection 
$$sSet(Y,S_{/p})\cong sSet_p(Y\star K,S),$$
where the subscript on right hand side indicates that we consider only those morphisms $f:Y\star K\rightarrow S$ such that $f|K=p$.
\end{prop}

\begin{prop}[Joyal]
	Let $X$ be an $\infty$-category and $K$ be a simplicial set. If $p:K\rightarrow X$ be a map of simplicial sets, then $X_{/p}$ is an $\infty$-category. Moreover, if $q:X\rightarrow Y$ is a categorical equivalence, then the induced map $X_{/p}\rightarrow Y_{/qp}$ is a categorical equivalence as well.  
\end{prop}

\begin{defin}[Slice $\infty$-Categorical \cite{DBLP:books/mk/Lurie}]
		Given $X$ be an $\infty$-category, $K$ be a simplicial set and $p:K\rightarrow X$ be a map of simplicial sets. Define the slice $\infty$-category $X_{/p}$ of the objects over $p$ or cones on $p$. Dually, $X_{p/}$ is the $\infty$-category of objects under $p$ or cocones on $p$.
\end{defin}

\begin{ejem}
	Let $X$ be an $\infty$-category and $x\in X$ be an object, which correspond to map  $x:\Delta^0\rightarrow X$. The objects of the $\infty$-category $X_{/x}$ of cones on $x$ are morphism $y\rightarrow x$ in $X$, and the morphism from $y\rightarrow x$ to $z\rightarrow x$ in $X_{/x}$, are the 2-simplexes 
			 	\[\xymatrix{
		& {y}\ar[dd]_{}\ar[rrdd]& &
		 \\ \\
		& z\ar[rr]_{}& & x
		& 
	}\]
in the $\infty$-category $X$. 
\end{ejem}

\subsection{Limits and colimits}

An object $t$ of a category is final if for each object $x$ in this category, there is a unique morphism $x\rightarrow t$. Next, one defines the final objects at $\infty$-categories, under a contractible space of morphisms.

\begin{defin}[Final object \cite{DBLP:books/mk/Lurie}]
An object $\omega\in X$ in an $\infty$-category $X$, is a final object if for any object $x\in X$, the Kan complex of morphisms $X(x,\omega)$ is contractible.
\end{defin}

\begin{corol}[\cite{DBLP:books/mk/Cisinski}]\label{corolally-final-objects}
	The final objects of an $\infty$-category $X$ form a Kan complex which is either empty or equivalent to the point. 
\end{corol}

\begin{corol}[\cite{DBLP:books/mk/Cisinski}]
	Let $x$ be a final object in an $\infty$-category $X$. For any simplicial set $A$, the constant map $A\rightarrow X$ with value $x$ is a final object in $Map(A,X)$.  
\end{corol}

\begin{defin}[Limit and colimit \cite{Joyal2002}]
	Let $X$ be an $\infty$-category and let $p:K\rightarrow X$ a map of simplicial sets. A colimit for $p$ is an initial object of $X_{p/}$, and a limit for $p$ is a final object of $X_{/p}$. 
\end{defin}

By the dual of Corollary \ref{corolally-final-objects}, if the colimit exists, then the Kan complex of initial objects is contractible, i.e., the initial object is unique up to contractible choice. 

\subsection{$\infty$-categories of presheaves}   

\begin{defin}[$\infty$-categories of presheaves \cite{DBLP:books/mk/Lurie}]
	Let $S$ be a simplicial set. One lets $P(S)$ or $PS$ denote simplicial set $Fun(S^{op},\mathscr{S})$; where $\mathscr{S}$ denotes the $\infty$-category of the small Kan complexes or $\infty$-groupoids, also called the $\infty$-category of spaces. One will say that $P(S)$ is the $\infty$-category of the presheaves on $S$.
\end{defin}

\begin{prop}[\cite{DBLP:books/mk/Lurie}]
Let $S$ be a simplicial set. The $\infty$-category $P(S)$ of the presheaves on $S$ admits all small limits and colimits.
\end{prop}

\begin{prop}[$\infty$-Categorical Yoneda Lemma \cite{DBLP:books/mk/Lurie}]
	Let $S$ be a simplicial set. Then the Yoneda embedding $j:S\rightarrow PS$ is fully faithful. 
\end{prop}

\begin{n}
	Let $X$ be an $\infty$-category and $S$ be a simplicial set. One lets $Fun^L(PS,X)$ denote the full subcategory of $Fun(PS,X)$ spanning those functors $PS\rightarrow X$ which preserve small colimits. 
	
\bigskip The motivation for this notation stems from Adjoint Functor Theorem (will be seen later), where $Fun^L(PS,X)$ also denotes the full subcategory of $Fun(PS,X)$ spanning those functors  which are left adjoints. 
\end{n}

\begin{teor}[\cite{DBLP:books/mk/Lurie}]
Let $S$ be a small simplicial set and let $X$ be an $\infty$-category which admits small colimits. The composition with the Yoneda embedding $j:S\rightarrow PS$ induces an equivalence of $\infty$-categories
$$Fun^L(PS,X)\rightarrow Fun(S,X).$$
\end{teor}

\subsection{Adjoint Functors}

Next is the definition of adjoint functors; motivated from classical definition of functors between categories.

\begin{defin}[Adjunction \cite{DBLP:books/mk/Cisinski}]
	Let $F:X\rightarrow Y$ and $G:Y\rightarrow X$ be functors between $\infty$-categories. One will say that $(F,G)$ form an adjoint pair, or that $F$ is the left adjoint of $G$, or that $G$ is the right adjoint of $F$, if there exists a functorial invertible map of the form
	$$\alpha_{x,y}:X(x,Gy)\rightarrow Y(Fx,y)$$
	in the $\infty$-category $\hat{\mathscr{S}}$ of all Kan complexes (not necessarily small), where the word functorial means that this map is the evaluation at $(x,y)$ of a morphism $\alpha$ in the $\infty$-category of the functors $Fun(X^{op}\times Y,\hat{\mathscr{S}})$. 
	
	\bigskip An adjunction from $A$ to $B$ is a triple $(F,G,\alpha)$, where $F$ and $G$ are the functors as above, while $\alpha$ is an invertible map from $X(-,G(-))$ to $Y(F(-),-)$ which exhibits $G$ as a right adjoint of $F$.
\end{defin}

\begin{prop}[\cite{DBLP:books/mk/Cisinski}]
	Let $F:X\rightarrow Y$ be a functor between small $\infty$-categories. For two adjunctions $(F,G,\alpha)$ and $(F,G',\alpha')$, the functorial maps $Gy\rightarrow G'y$ which are compatibles with $\alpha$ and $\alpha'$ form a contractible Kan complex (i.e., there is unique way to identify $G$ with $G'$ as right adjoint of $F$).
\end{prop}

\begin{prop}[\cite{DBLP:books/mk/Lurie}]
		Let $F:X\rightarrow Y$ and $G:Y\rightarrow X$ be functors between $\infty$-categories. The following conditions are equivalent:
		\begin{enumerate}
			\item The functor $F$ is a left adjoint to $G$.
			\item There is a unit transformation $u:id_X\rightarrow G\circ F$
		\end{enumerate}
\end{prop}

\begin{prop}[\cite{DBLP:books/mk/Lurie}]
	Let $F:X\rightarrow Y$ be a functor between $\infty$-categories which has a right adjoint $G:Y\rightarrow X$. Then $F$ preserves all the colimits which exist in $X$, and $G$ preserves all the limits which exist in $Y$.  
\end{prop}

\begin{defin}[\cite{DBLP:books/mk/Lurie}]
	Let $X$ and $Y$ be $\infty$-categories. One lets $Fun^L(X,Y)\subseteq Fun(X,Y)$ denote the full subcategory of $Fun(X,Y)$ spanned by those functors $F:X\rightarrow Y$ which are left adjoints. Similarly, one will define $Fun^R(X,Y)$ to be the full subcategory of $Fun(X,Y)$ spanned by those functors which are right adjoints. 
\end{defin}

\begin{prop}[\cite{DBLP:books/mk/Lurie}]
Let $X$ and $Y$ be $\infty$-categories. Then the $\infty$-categories $Fun^L(X,Y)$ and $Fun^R(X,Y)^{op}$ are (canonically) equivalent. 
\end{prop}

\subsection{Filtered $\infty$-Categories and Compact Objects}

Recall that the filtered categories are the generalization of the partially ordered set $A$ which are filtered, i.e., those that satisfy the condition: every finite subset of $A$ has an upper bound in $A$. Next is the definition of filtered $\infty$-category which generalizes the  classical version to categories.   

\begin{defin}[$\kappa$-Filtered \cite{DBLP:books/mk/Lurie}]
	Let $\kappa$ be a regular cardinal and let $X$ be an $\infty$-category. We say that $X$ is $\kappa$-filtered if, for every $\kappa$-small set simplicial $S$ and all map $f:S\rightarrow X$ there is a map $\overline{f}:S^\rhd\rightarrow X$ (with $S^\rhd=S\star\Delta^0$) which extends to $f$. (i.e., $X$ is filtered if this the extension property with respect to the inclusion $S\subseteq S^\rhd$ for every $\kappa$-small simplicial set $S$). One will say that $X$ is filtered if this is $\omega$-filtered.
\end{defin}

\begin{defin}[Weakly contractible]
	A $\infty$-category $X$ is weakly contractible if its geometric realization $|X|\in Top$ is contractible. 
\end{defin}

\begin{prop}[\cite{DBLP:books/mk/Lurie}]
	Let $X$ be a filtered $\infty$-category. Then $X$ is weakly contractible.
\end{prop}

\begin{defin}[\cite{DBLP:books/mk/Lurie}]
Let $\kappa$ be a regular cardinal and $X$ be an $\infty$-category.	One will say that $X$ admits  $\kappa$-filtered colimits if it admits colimits of all indexed diagrams on any $\kappa$-filtered $\infty$-category. 
\end{defin}

\begin{defin}[Continuity \cite{DBLP:books/mk/Lurie}]
	Let $X$ be an $\infty$-category which admits $\kappa$-filtered colimits. One will say a functor $f:X\rightarrow Y$ is $\kappa$-continuous if it preserves $\kappa$-filtered colimits.
\end{defin}

\begin{defin}[$\kappa$-Compact \cite{DBLP:books/mk/Lurie}]
Let $X$ be an $\infty$-category, $x\in X$ and let $X(x,-):X\rightarrow\hat{\mathscr{S}}$ be the corepresented by $x$. If $X$ admits $\kappa$-filtered colimits, then one will say that $x$ is $\kappa$-compact if the functor corepresented by it is $\kappa$-continuous. One will say that $x$ is compact if it is $\omega$-compact (and $X$ admits filtered colimits).  
\end{defin}

\begin{rem}
	Let $X$ be an accessible $\infty$-category  and $\kappa$ be a regular cardinal. The full subcategory $X^\kappa\subseteq X$ consisting of all the $\kappa$-compact objects of $X$ is essentially small, i.e., there exists a small $\infty$-category $Y$ equivalent to $X^\kappa$.
\end{rem}

\begin{n}
	Let $X$ be an $\infty$-category and $\kappa$ be regular  cardinal. Denote by $Ind_\kappa(X)$ the closure of $X$ under $\kappa$-filtered colimits.
\end{n}

\begin{prop}[\cite{DBLP:books/mk/Lurie}]
	Let $X$ be a small $\infty$-category and $\kappa$ a regular cardinal. The $\infty$-category $P^\kappa(X)$ of $\kappa$-compact objects of $P(X)$ is essentially small:
	that is, there exists a small $\infty$-category $Y$ and an equivalence $i : Y\rightarrow P^\kappa(X)$.
	Let $F : Ind_\kappa(Y)\rightarrow P(X)$ be a $\kappa$-continuous functor such that the composition of $F$ with the Yoneda embedding
	$$Y\rightarrow Ind_\kappa(Y)\rightarrow P(X)$$
	is equivalent to $i$. Then $F$ is an equivalence of $\infty$-categories.
\end{prop}

\subsection{Accessible $\infty$-category}

In this section, we define everything related to accessibility, necessary to define presentable $\infty$-categories.

\begin{defin}[Accessible $\infty$-Category \cite{DBLP:books/mk/Lurie}]
Let $\kappa$ be a regular cardinal. An $\infty$-category $X$ is $\kappa$-accessible if there exists a small $\infty$-category $X^0$ and an equivalence
$$Ind_\kappa(X^0)\rightarrow X.$$
One will say that $X$ is accessible if it is $\kappa$-accessible for some regular cardinal
$\kappa$.
\end{defin}

\begin{defin}[Accessible Functor \cite{DBLP:books/mk/Lurie}]
Let $X$ be an accessible $\infty$-category, then a functor $F :X\rightarrow X'$ is accessible if it is $\kappa$-continuous for some regular cardinal $\kappa$.	
\end{defin}

\begin{ejem}
The $\infty$-category $\mathscr{S}$ of spaces is accessible. More generally,
for any small $\infty$-category $X$, the $\infty$-category $P(X)$ is accessible.
\end{ejem}

\begin{n}
	Denote by $Cat_\infty$ the $\infty$-category of all small $\infty$-categories, and by $CAT_\infty$ the $\infty$-category of the all the $\infty$-categories.
\end{n}

\begin{defin}[\cite{DBLP:books/mk/Lurie}]
	Let $\kappa$ be a regular cardinal. We let $Acc_\kappa\subseteq CAT_\infty$
	denote the subcategory defined as follows:
	\begin{enumerate}
		\item The objects of $Acc_\kappa$ are the $\kappa$-accessible $\infty$-categories.
		\item A functor $F : X\rightarrow Y$ between accessible $\infty
		$-categories belongs to $Acc_\kappa$
		if and only if $F$ is $\kappa$-continuous and preserves $\kappa$-compact objects.
	\end{enumerate}
Let $Acc = \bigcup_\kappa Acc_\kappa$ . We will refer to $Acc$ as the $\infty$-category of accessible $\infty$-categories.
\end{defin}

\subsection{Presentable $\infty$-Categories}

Next, we define presentable $\infty$-categories and their characterizations.

\begin{defin}[Presentable $\infty$-category \cite{DBLP:books/mk/Lurie}]
	An $\infty$-category $X$ is presentable if $X$ is accessible and admits small colimits.
\end{defin}

\begin{teor}[\cite{Simpson99} and \cite{DBLP:books/mk/Lurie}]
	Let $X$ be an $\infty$-category. The following
	conditions are equivalent:
	\begin{enumerate}
		\item The $\infty$-category $X$ is presentable.
		\item The $\infty$-category $X$ is accessible, and for every regular cardinal $\kappa$ the full subcategory $X^\kappa$ admits $\kappa$-small colimits.
		\item There exists a regular cardinal $\kappa$ such that $X$ is $\kappa$-accessible and $X^\kappa$ admits $\kappa$-small colimits.
		\item There exists a regular cardinal $\kappa$, a small $\infty$-category $Y$ which admits $\kappa$-small colimits, and an equivalence $Ind_\kappa X\rightarrow Y$.
	\end{enumerate}
\end{teor}

\begin{prop}[\cite{DBLP:books/mk/Lurie}]
Let $F : X\rightarrow Y$ be a functor between presentable $\infty$-categories. Suppose that $X$ is $\kappa$-accessible. The following conditions are
equivalent:	
\begin{enumerate}
	\item The functor $F$ preserves small colimits.
	\item The functor $F$ is $\kappa$-continuous, and the restriction $ F|X^\kappa$ preserves $\kappa$-small colimits.
\end{enumerate}
\end{prop}

\begin{teor}[Adjoint Functor Theorem \cite{DBLP:books/mk/Lurie}] 
	Let $F : X\rightarrow Y$ be a functor
	between presentable $\infty$-categories.
	\begin{enumerate}
		\item The functor $F$ has a right adjoint if and only  it preserves small
		colimits.
		\item 	The functor $F$ has a left adjoint if and only if it is accessible and
		preserves small limits.
	\end{enumerate}
\end{teor}

\begin{defin}
	Define the subcategories $\mathscr{P}r^L,\mathscr{P}r^R\subseteq CAT_\infty$ as follows:
	\begin{enumerate}
		\item The objects of both $\mathscr{P}r^L$ and $\mathscr{P}r^R$ are the presentable $\infty$-categories.
		\item A functor $F : X\rightarrow Y$ between presentable $\infty$-categories is a morphism in $\mathscr{P}r^L$ if and only if $F$ preserves small colimits. Hence $F\in Fun^L(X,Y)$.
		\item A functor $G : X\rightarrow Y$ between presentable $\infty$-categories is a morphism in $\mathscr{P}r^R$ if and only if $G$ is accessible and preserves small limits. Thus $G\in Fun^R(X,Y)$.
	\end{enumerate}
\end{defin}

\begin{prop}[\cite{DBLP:books/mk/Lurie}]
The $\infty$-categories $\mathscr{P}r^L$ and $\mathscr{P}r^R$ admits all small limits, and the inclusion functors $\mathscr{P}r^L,\mathscr{P}r^R\subseteq CAT_\infty$ preserve all small limits. 	
\end{prop}

\begin{prop}[\cite{DBLP:books/mk/Lurie}]\label{prop-closed-colimits}
	Let $X$ be an presentable $\infty$-category and let $S$ be a small simplicial set. Then $Fun(S, X)$ and $Fun^L(S, X)$ are presentable.
\end{prop}

\begin{rem}[\cite{DBLP:books/mk/Lurie}]
	The presentable $\infty$-category $Fun^L(X, Y)$ can be regarded as an internal mapping object in $Pr^L$. That is, there exists an operation $\otimes$ that endows $Pr^L$ with the structure of a symmetric monoidal $\infty$-category. Proposition \ref{prop-closed-colimits} can be interpreted as asserting that this monoidal structure is closed.
\end{rem}

\subsection{Compactly Generated $\infty$-Categories}

\begin{defin}[\cite{DBLP:books/mk/Lurie}]
	Let $\kappa$ be a regular cardinal. We will say that an $\infty$-category $X$ is $\kappa$-compactly generated if it is presentable and $\kappa$-accessible.
	When $\kappa = \omega$, we will simply say that $X$ is compactly generated.
\end{defin}

\begin{prop}[\cite{DBLP:books/mk/Lurie}]
		Let $\kappa$ be a regular cardinal and let $F:X \rightleftarrows Y: G$ be a
		pair of adjoint functors, where $X$ and $Y$ admit small $\kappa$-filtered colimits and these are $\kappa$-compactly generated.
		\begin{enumerate}
			\item If $G$ is $\kappa$-continuous, then $F$ carries $\kappa$-compact objects of $X$ to $\kappa$-
			compact objects of $Y$.
			\item Conversely, if $X$ is $\kappa$-accessible and $F$ preserves $\kappa$-compactness, then
			$G$ is $\kappa$-continuous.
		\end{enumerate}
\end{prop}

\begin{defin}
	If $\kappa$ is a regular cardinal, define the subcategories $\mathscr{P}r^L_\kappa,\mathscr{P}r^R_\kappa\subseteq CAT_\infty$ as follows
	\begin{enumerate}
		\item The objects of both $\mathscr{P}r^L_\kappa$ and $\mathscr{P}r^R_\kappa$ are $\kappa$-compactly generated $\infty$-categories.
		\item The morphisms in $\mathscr{P}r^R_\kappa$ are $\kappa$-continuous limit-preserving functors.
		\item The morphisms in $\mathscr{P}r^L_\kappa$ are functors which preserve small colimits and $\kappa$-
		compact objects. 
	\end{enumerate} 
\end{defin}

\begin{prop}
	The $\infty$-category $\mathscr{P}r^R_\kappa$ admits all the limits and the inclusion $\mathscr{P}r^L_\kappa\subseteq CAT_\infty$ preserves small limits. 
\end{prop}

\begin{defin}[\cite{DBLP:books/mk/Lurie}]
	Define the subcategory $CAT^{Rex(\kappa)}_\infty$ whose objects are (not necessarily small) $\infty$-categories which admit $\kappa$-small colimits and whose morphisms are functors which preserve $\kappa$-small colimits, and let $Cat^{Rex(\kappa)}_\infty=CAT^{Rex(\kappa)}_\infty\cap Cat_\infty$. 
\end{defin}

\begin{prop}[\cite{DBLP:books/mk/Lurie}]
	Let $\kappa$ be a regular cardinal and let
	$$\theta:\mathscr{P}r^L_\kappa\rightarrow CAT^{Rex(\kappa)}_\infty$$
	be the functor which associates to a $\kappa$-compactly generated $\infty$-category $X$ the full subcategory $X^\kappa\subseteq X$ spanned by the $\kappa$-compact
	objects of $X$. Then the functor $\theta$ is fully faithful.
\end{prop}

\subsection{Kleisli $\infty$-categories} 

Next, the definition of Kleisli structures on the $\infty$-categories \cite{MartinezHoDT}; a general and direct version of those initially introduced by \cite{Hyland14} for the case of bicategories.

\begin{defin}[Kleisli structure \cite{MartinezHoDT}]\label{Kleisli-strucutre}
	Let $\mathcal{K}$ be an $\infty$-category and $\mathcal{A}$ be an $\infty$-subcategory of $\mathcal{K}$. A Kleisli structure $P$ on $\mathcal{A}\subseteq \mathcal{K}$ is the following.
	\begin{itemize}
		\item For each vertex $a\in\mathcal{A}$ an arrow $y_a:a\rightarrow Pa$ in $\mathcal{K}$.
		\item For each $a,b \in\mathcal{A}$ a functor
		$$\mathcal{K}(a,Pb)\rightarrow\mathcal{K}(Pa,Pb),\hspace{0.5cm}f\mapsto f^{\#}.$$
		\item A subcategory $\mathcal{K}^{L}\subseteq\mathcal{K}$ which sets, for all the vertices $a,b\in\mathcal{A}$, the homotopy equivalence
		$$\xymatrix{
			\mathcal{K}(a,Pb)\ar@<1ex>[r]^{(-)^\#} & \mathcal{K}^L(Pa,Pb).\ar@<1ex>[l]^{(-)y_a}
		}$$
		
		such for each horn $(g^\#,-,f^\#):\Lambda_1^2\rightarrow\mathcal{K}^L$  one has the equality of fibres 
		$$F_{g^\#,f^\#}^L=F_{g^\#,f^\#},$$
		where $F_{g^\#,f^\#}^L$ and $F_{g^\#,f^\#}$ are the fibres of the canonical maps $Fun(\Delta^2,\mathcal{K}^L)\rightarrow Fun(\Lambda^2_1,\mathcal{K}^L)$ and $Fun(\Delta^2,\mathcal{K})\rightarrow Fun(\Lambda^2_1,\mathcal{K})$ respectively, with $f:a\rightarrow Pb$ and $g:b\rightarrow Pc$ be edges in $\mathcal{K}$.

	\end{itemize}   
\end{defin}

It is clear that $P$ is a functor from  $\mathcal{A}$ to $\mathcal{K}$ such that for each 1-simplex $f:a\rightarrow b$ of $\mathcal{A}$,  sets $Pf=(y_bf)^{\#}:Pa\rightarrow Pb$.

\begin{ejem}
	The functor $P:Cat_\infty\rightarrow CAT_\infty$, given by $PA=[A^{op},\mathscr{S}]$, is a Kleisli structure on $Cat_\infty\subseteq CAT_\infty$, where $\mathcal{K}^L=\mathscr{P}r^L.$
	
	\medskip We have the categorical equivalence 
	$$\xymatrix{
		Fun(A,PB)\ar@<1ex>[r]^{(-)^\#} & Fun^L(PA,PB),\ar@<1ex>[l]^{(-)y_A}
	}$$
	where $Fun^L(PA,PB)$ is the $\infty$-category of functors which preserve small colimits. Restricting
	to the subcategory $\mathscr{P}r^L\subseteq CAT_\infty$, whose objects are presentable $\infty$-categories and morphisms are functors which preserve small colimits, then the categorical equivalence is restricted to the homotopy equivalence  
	$$\xymatrix{
		CAT_\infty(A,PB)\ar@<1ex>[r]^{(-)^\#} & \mathscr{P}r^L(PA,PB).\ar@<1ex>[l]^{(-)y_A}
	}$$

\end{ejem}

\begin{defin}[Kleisli $\infty$-category \cite{MartinezHoDT}]
	Given a Kleisli structure $P$ on $\mathcal{A}\subseteq \mathcal{K}$. We define $Kl(P)$ as the simplicial set embedded in $\mathcal{K}^L$, where the objects  of $Kl(P)$ are the objects of $\mathcal{A}$, the morphism spaces is defined by $$Kl(P)(a,b):=\mathcal{K}(a,Pb).$$
	for the all objects $a,b\in\mathcal{A}$. The embedding $Kl(P)\hookrightarrow \mathcal{K}^L$ is induced by definition of $Kl(P)(a,b)$ and the homotopy equivalence $\mathcal{K}(a,Pb)\rightarrow\mathcal{K}^L(Pa,Pb)$ of the third item in the Definition \ref{Kleisli-strucutre}. Thus, all the n-simplexes in $Kl(P)$ are defined by the inverse image (of the embedding) of the n-simplexes in $\mathcal{K}^L$. 
\end{defin}

\begin{rem}
	
By Definition \ref{Kleisli-strucutre}, the homotopy equivalence $Kl(P)(a,b)\rightarrow\mathcal{K}^L(Pa,Pb)$ gets to establish that $Kl(P)$ is the $\infty$-category embedded in $\mathcal{K}^L$. Another interesting way would be to define $Kl(P)$ as a weighted colimit or the pushout of diagram $$(Cat_\infty\subseteq CAT_\infty\xleftarrow{P}Cat_\infty)$$
	in the category of simplicial sets $Set_{\Delta}=Fun(\Delta^{op},Set)$ (complete and cocomplete) and so $Kl(P)$ would be an $\infty$-category.
\end{rem}

\begin{prop}[\cite{MartinezHoDT}]
	$P(A\times B)\simeq PA\otimes PB$ in the $\infty$-category $\mathscr{P}r^L$.
\end{prop}

\subsection{Complete homotopy partial order (c.h.p.o)}

Finally, to close this section on fundamentals, we present the definitions of homotopy partial order, and complete homotopy partial order.

\begin{defin}[h.p.o. \cite{MartinezHoDT} and \cite{Martinez21}]\label{h.p.o}
	Let $\hat{K}$ be an $\infty$-category. The largest Kan complex $K\subseteq\hat{K}$ is a homotopy partial order (h.p.o), if for every $x,y\in K$ one has that $\hat{K}(x,y)$ is contractible or empty. The Kan complex $K$ admits a relation of h.p.o $\precsim$ defined for each $x,y\in K$ as follows:
	$x\precsim y$ if $\hat{K}(x,y)\neq\emptyset$, thus, the pair $(K,\precsim)$ is a h.p.o. (we denote only $K$). 
\end{defin}

\begin{defin}[c.h.p.o. \cite{MartinezHoDT} and \cite{Martinez21}] 
	Let $K$ be a h.p.o.
	\begin{enumerate}
		\item A h.p.o $X\subseteq K$ is directed if $X\neq \emptyset$ and for each $x,y\in X$, there exists $z\in X$ such that $x\precsim z$ and $y\precsim z$.
		\item $K$ is a complete homotopy partial order (c.h.p.o) if
		\begin{enumerate}
			\item There are initial objects, i.e.,  $0\in K$ is an initial object if for each $x\in K$, $0\precsim x$. 
			\item For each directed $X\subseteq\mathcal{K}$ the supremum (or colimit) $\bigcurlyvee X\in\mathcal{K}$ exists. 
		\end{enumerate}
	\end{enumerate}
\end{defin}

\begin{rem}\label{Remark}
	If $K$ is a c.h.p.o, note that $\hat{K}$ is weakly contractible (i.e., its geometric realization $|\hat{K}|$ is contractible) but $K$ not necessarily is contractible: by example if $\hat{K} = \langle Sing(|\partial\Delta^n|)\cup \{0\}\rangle$ is the smallest $\infty$-category that contains to the Kan complex $K=Sing(|\partial\Delta^n|)\cup \{0\}$ such that $0$ be the initial object unique of $\hat{K}$, where $Sing: Top \rightarrow Kan$  be the singular functor from category of  topological spaces to the category of Kan complexes and $|\,\,|:Kan \rightarrow Top$ be the geometric realization. Since, $Sing(|\partial\Delta^n|)$ is isomorphic to the sphere $S^{n-1}$, we denoted to $Sing(|\partial\Delta^n|)$ as $S^{n-1}$. 
\end{rem}

\section{Homotopy Domain Equation on an arbitrary Cartesian closed $\infty$-category}\label{HDE arbitrary cci}

This section is a direct generalization of the traditional methods for solving domain equations in Cartesian closed categories (see \cite{DBLP:books/mk/Asperti91} and \cite{DBLP:books/mk/Abramsky94}), in the sense of obtaining solutions for certain types of equations, which we call Homotopy Domain Equations in any Cartesian closed $\infty$-category.

\begin{defin}\cite{DBLP:books/mk/Asperti91}
	\begin{enumerate}
		\item An $\omega$-diagram in an $\infty$-category $\mathcal{K}$ is a diagram with the following structure:
		$$K_0 \stackrel{f_0}{\longrightarrow} K_1\stackrel{f_1}{\longrightarrow} K_2\longrightarrow\cdots \longrightarrow K_n \stackrel{f_n}{\longrightarrow}K_{n+1}\longrightarrow\cdots$$
		(dually, one defines $\omega^{op}$-diagrams by just reversing the arrows).
		\item An $\infty$-category $\mathcal{K}$ is $\omega$-complete ($\omega$-cocomplete) if it has limits (colimits) for all $\omega$-diagrams.
		\item A functor $F:\mathcal{K}\rightarrow\mathcal{K}$ is $\omega$-continuous if it preserves (under equivalence) all colimits of $\omega$-diagrams.  
	\end{enumerate}
\end{defin}

\begin{teor}\label{point-fixed-theorem}
	Let $\mathcal{K}$ be an $\infty$-category. Let $F:\mathcal{K}\rightarrow\mathcal{K}$ be a $\omega$-continuous (covariant) functor and take a vertex $K_0\in\mathcal{K}$ such that there is an edge $\delta\in\mathcal{K}(K_0,FK_0)$. Assume also that $(K,\{\delta_{i,\omega}\in \mathcal{K}(F^iK_0,K)\}_{i\in\omega})$ is a colimit for the $\omega$-diagram $(\{F^iK_0\}_{i\in\omega},\{F^i\delta\}_{i\in\omega})$, where $F^0K_0=K_0$ and $F^0\delta=\delta$. Then $K\simeq FK$.
\end{teor}
\begin{proof}[Proof]
	We have that $(FK,\{F\delta_{i,\omega}\in \mathcal{K}(F^{i+1}K_0,FK)\}_{i\in\omega})$ is a colimit for 
	$$(\{F^{i+1}K_0\}_{i\in\omega},\{F^{i+1}\delta\}_{i\in\omega})$$and $(K,\{\delta_{i+1,\omega}\in \mathcal{K}(F^{i+1}K_0,K)\}_{i\in\omega})$ is a cocone for the same diagram. Then, there is a unique edge (under homotopy; the space of choices is contractible) $h:FK\rightarrow K$ such that  $h.F\delta_{i,\omega}\simeq\delta_{i+1,\omega}$, for each $i\in\omega$. We add to $(FK,\{F\delta_{i,\omega}\in \mathcal{K}(F^{i+1}K_0,FK)\}_{i\in\omega})$ the edge $F\delta_{0,\omega}.\delta\in\mathcal{K}(K_0,FK)$. This gives a cocone for $(\{F^iK_0\}_{i\in\omega},\{F^i\delta\}_{i\in\omega})$ and, since $(K,\{\delta_{i,\omega}\in \mathcal{K}(F^{i}K_0,K)\}_{i\in\omega})$ is its colimit, there is a unique edge (under homotopy) $k:K\rightarrow FK$ such that $k.\delta_{i+1,\omega}\simeq F\delta_i$  and $k.\delta_{0,\omega}\simeq F\delta_{0,\omega}.\delta$, for each $i\in\omega$. But, $h.k.\delta_{i+1,\omega}\simeq h.F\delta_{i,\omega}\simeq\delta_{i+1,\omega}$ and $h.k.\delta_{0,\omega}\simeq\delta_{0,\omega}$, for each $i\in\omega$, thus $h.k$ is a mediating edge between the colimit $(K,\{\delta_{i,\omega}\in \mathcal{K}(F^iK_0,K)\}_{i\in\omega})$ and itself (besides $I_K$). Hence, by unicity (under homotopy) $h.k\simeq I_K$. In the same way, we prove that $k.h\simeq I_{FK}$, and we conclude that $F$ has a fixed point.
\end{proof}

\begin{defin}
	An $\infty$-category $\mathcal{K}$ is a $0$-$\infty$-category if
	\begin{enumerate}
		\item every Kan complex $\mathcal{K}(A,B)$ is a c.h.p.o, whose initial object is noted by $0_{A,B}$, 
		\item composition of morphisms is a continuous operation with respect to the homotopy order,
		\item for every $f$ in $\mathcal{K}(A,B)$, $0_{B,C}.f\simeq 0_{A,C}$.
	\end{enumerate}
\end{defin}

\begin{ejem}
	The $\infty$-category $Kl(P)$ \cite{Martinez2HoDTvsHoTT21}, is an $0$-$\infty$-category. See proof on the Section \ref{HDE on Kl(P)} of the present paper. 
\end{ejem}

\begin{defin}[h-projection par]\label{h-projection}
	Let $\mathcal{K}$ be a $0$-$\infty$-category, and let $f^+:A\rightarrow B$ and $f^-:B\rightarrow A$ be two morphisms in $\mathcal{K}$. Then $(f^+,f^ -)$ is a homotopy projection (or h-projection) pair (from $A$ to $B$) if $f^-.f^+\simeq I_A$ and $f^+.f^-\precsim I_B$. If $(f^+,f^-)$ is an h-projection pair, $f^+\in\mathcal{K}^{HE}(A,B)$ is an h-embedding (homotopy embedding) and $f^-\in\mathcal{K}^{HP}(A,B)$ is an h-projection (homotopy projection). Where $\mathcal{K}^{HE}$ is the subcategory of $\mathcal{K}$ with the same objects and the h-embeddings as morphisms, and $\mathcal{K}^{HP}$ is the subcategory of $\mathcal{K}$ with the same objects and the h-projections as morphisms.               
\end{defin}

\begin{defin}[h-projections pair $0$-$\infty$-category.]\label{h-projection-inftycategory}
	Let $\mathcal{K}$ be a $0$-$\infty$-category. The  $0$-$\infty$-category $\mathcal{K}^{HPrj}$ is the $\infty$-category embedding in $\mathcal{K}^{HE}$ with the same objects of $\mathcal{K}$ and h-projection pairs $(f^+,f^-)$ as morphisms.
\end{defin}

\begin{rem}
	Every h-embedding $i$ has unique (under homotopy) associated h-projection  $j= i^R$ (and, conversely, every h-projection $j$ has a unique (under homotopy) associated h-embedding $i=j^L$), since if there is $j_0$ such that $j_0.i\simeq I$ and $i.j_0\precsim I$ (under homotopy), so $j_0\succsim j_0.i.j\simeq j$ and $j_0\precsim j$ (under homotopy). Thus, $j_0\simeq j$ under homotopy (and, in the same way, we have $i_0\simeq i$ under homotopy). $\mathcal{K}^{HPrj}$ and  is equivalent to a subcategory  $\mathcal{K}^{HE}$ of $\mathcal{K}$ that has h-embeddings as morphisms  (as well to a subcategory $\mathcal{K}^{HP}$ of $\mathcal{K}$ which has h-projections as morphisms).  
\end{rem}

\begin{defin}
	Given a $0$-$\infty$-category $\mathcal{K}$, and a contravariant functor in the first component $F:\mathcal{K}^{op}\times\mathcal{K}\rightarrow\mathcal{K}$, the functor covariant $F^{+-}:\mathcal{K}^{HPrj}\times\mathcal{K}^{HPrj}\rightarrow\mathcal{K}^{HPrj}$ is defined by
	\begin{align*}
		&F^{+-}(A,B)=F(A,B), \\
		&F^{+-}((f^+,f^-),(g^+,g^-))=(F(f^-,g^+),F(f^+,g^-)),
	\end{align*}
	where $A,B$ are vertices and $(f^+,f^-),(g^+,g^-)$ are n-simplexes pairs in $\mathcal{K}^{HPrj}$.
\end{defin}

Given the $\omega$-chain $(\{K_i\}_{i\in\omega},\{f_i\}_{i\in\omega})$ in an h-projective $0$-$\infty$-category. Let $(K,\{\gamma_i\}_{i\in\omega})$ be a limit for $(\{K_i\}_{i\in\omega},\{f_i^-\}_{i\in\omega})$ in $\mathcal{K}$. Note that $\delta_i.\gamma_i\simeq\delta_{i+1}.f_i^+.f_i^-.\gamma_{i+1} \\\precsim\delta_{i+1}.\gamma_{i+1}$, for each $i\in\omega$. Then, $\{\delta_{i}.\gamma_{i}\}$ is an $\omega$-chain and its colimit is $\Theta=\bigcurlyvee_{i\in\omega}\{\delta_{i}.\gamma_{i}\}$. Now for each $j\in\omega$ one has
\begin{align*}
	\gamma_j.\Theta_j&=\gamma_j.\bigcurlyvee_{i\in\omega}\{\delta_{i}.\gamma_{i}\} \\
	&\simeq\gamma_j.\bigcurlyvee_{i\geq j}\{\delta_{i}.\gamma_{i}\} \\
	&\simeq\bigcurlyvee_{i\geq j}\{(\gamma_j.\delta_{i}).\gamma_{i}\} \\
	&\simeq\bigcurlyvee_{i\geq j}\{f_{i,j}.\gamma_{i}\} \\
	&\simeq\gamma_j.
\end{align*}
Thus, $\Theta$ is a mediating edge between the limit $(K,\{\gamma_i\}_{i\in\omega})$ for $\omega^{op}$-diagram $(\{K_i\}_{i\in\omega},\{f_i^-\}_{i\in\omega})$  and itself (besides $I_K$). So, by unicity (under equivalence) $\Theta\simeq I_K$. This result guarantees the proof of the following theorems of this section. All these proofs are similar to case of the 0-categories (except for uniqueness proofs, which are under homotopy, i.e., under contractible spaces to ensure a weak uniqueness, 
examine, for example, the final part of the proof of Theorem  \ref{colimit-HPrj-category}) and, for that, the reader is referred to \cite{DBLP:books/mk/Asperti91}.   

\begin{teor}\label{colimit-HPrj-category}
	Let $\mathcal{K}$ be a  $0$-$\infty$-category. Let $(\{K_i\}_{i\in\omega},\{f_i\}_{i\in\omega})$ be  an $\omega$-diagram in $\mathcal{K}^{HPrj}$. If $(K,\{\gamma_i\}_{i\in\omega})$ is a limit for $(\{K_i\}_{i\in\omega},\{f_i^-\}_{i\in\omega})$ in $\mathcal{K}$, then $(K,\{(\delta_i,\gamma_i)\}_{i\in\omega})$ is a colimit for $(\{K_i\}_{i\in\omega},\{f_i\}_{i\in\omega})$ in $\mathcal{K}^{HPrj}$ (that is, every $\gamma_i$ is a right member of a projection pair). 
\end{teor}
\begin{proof}[Proof]
	Fix $K_j$. For each $i$ define $f_{j,i}:K_j\rightarrow K_i$ by:
	
	$$f_{j,i}=\left\lbrace\begin{array}{c} f_i^-.f_{i+1}^-\ldots f_{j-1}^-~~if~~i<j,\\ I_{K_j}\hspace{2.2cm}if~~i=j, \\
	f_{i-1}^+\ldots f_{j+1}^+.f_j^+~~if~~i>j.\end{array}\right.$$
	
	$(K_i,\{f_{j,i}\}_{i\in\omega})$ is a cone for $(\{K_i\}_{i\in\omega},\{f_i^-\}_{i\in\omega})$, since $f_i^-.f_{j,i+1}\simeq f_{j,i}$. Thus there is a unique morphism (under homotopy) $\delta_j:K_j\rightarrow K$ such that $\gamma_i.\delta_j\simeq f_{j,i}$ for each $i\in\omega$. If $i=j$, $\gamma_j.\delta_j\simeq I_{K_j}$. 
	
	\medskip Since $\Theta=\bigcurlyvee_{i\in\omega}\{\delta_{i}.\gamma_{i}\}\simeq I_K$, $\delta_i.\gamma_i\precsim I_K$ for each $i\in\omega$. Thus, $(\delta_i,\gamma_i)$ is an h-projection pair for each $i\in\omega$.
	
	\medskip One still has to check that $(f_j^+,f_j^-).(\delta_{j+1},\gamma_{j+1})\simeq(\delta_{j},\gamma_{j})$. We have that $f_j^-.\gamma_{j+1}\simeq\gamma_{j}$ by the definition of a cone in $\mathcal{K}$. In order of to prove that $\delta_{j+1}.f_j^+\simeq\delta_j$, note that $\gamma_i.(\delta_{j+1}.f_j^+)\simeq f_{j+1,i}.f_j^+\simeq f_{j,j}\simeq I_{K_j}\simeq\gamma_i.\delta_j$, by unicity (under homotopy) of $\delta_j:K_j\rightarrow K$, $\delta_{j+1}.f_j^+\simeq\delta_j$.	Thus,  $(K,\{(\delta_i,\gamma_i)\}_{i\in\omega})$ is a cone in $\mathcal{K}^{HPrj}$.  
	
	\medskip One proves next that $(K,\{(\delta_i,\gamma_i)\}_{i\in\omega})$ is a colimit. Let $(K',\{(\delta_i',\gamma_i')\}_{i\in\omega})$ be another cocone for $(\{K_i\}_{i\in\omega},\{f_i\}_{i\in\omega})$. That is, for each $i\in\omega$:
	\begin{align*}
		&  \delta_i'.\gamma_i\simeq\delta_{i+1}'.f_i^+.f_i^-\gamma_{i+1}\precsim\delta_{i+1}'.\gamma_{i+1}  \\
		& \delta_i.\gamma_i'\simeq\delta_{i+1}.f_i^+.f_i^-\gamma_{i+1}'\precsim\delta_{i+1}.\gamma_{i+1}'. 
	\end{align*}
	Define thus:
	\begin{align*}
		& h=\bigcurlyvee_{i\in\omega}\{\delta_i'.\gamma_i\}:K\rightarrow K'\\
		&k=\bigcurlyvee_{i\in\omega}\{\delta_i.\gamma_i'\}:K'\rightarrow K.
	\end{align*}
	Observe that $(h,k)$ is an h-projection pair, since:	
	\begin{align*}
		k.h&=\bigcurlyvee_{i\in\omega}\{\delta_{i}.\gamma_{i}'\}.\bigcurlyvee_{i\in\omega}\{\delta_{i}'.\gamma_{i}\} \\
		&\simeq\bigcurlyvee_{i\in\omega}\{\delta_{i}.(\gamma_{i}'.\delta_{i}').\gamma_{i}\} \\
		&\simeq\bigcurlyvee_{i\in\omega}\{\delta_i.\gamma_{i}\} \\
		&\simeq\Theta=I_K
	\end{align*}
	and
	\begin{align*}
		h.k&=\bigcurlyvee_{i\in\omega}\{\delta_{i}'.\gamma_{i}\}.\bigcurlyvee_{i\in\omega}\{\delta_{i}.\gamma_{i}'\} \\
		&\simeq\bigcurlyvee_{i\in\omega}\{\delta_{i}'.(\gamma_{i}.\delta_{i}).\gamma_{i}'\} \\
		&\simeq\bigcurlyvee_{i\in\omega}\{\delta_i'.\gamma_{i}'\} \\
		&\precsim I_{K'}
	\end{align*}
	Moreover, $(h,k)$ is a mediating morphism between $(K,\{(\delta_i,\gamma_i)\}_{i\in\omega})$ and $(K',\{(\delta_i',\gamma_i')\}_{i\in\omega})$, since for each $i\in\omega$:
	\begin{align*}
		(h,k).(\delta_j,\gamma_j)&=(h.\delta_j\,,\,\gamma_j.k) \\
		&\simeq(\bigcurlyvee_{i\in\omega}\{\delta_{i}'.\gamma_{i}\}.\delta_j\,,\,\gamma_j.\bigcurlyvee_{i\in\omega}\{\delta_{i}.\gamma_{i}'\}) \\
		&\simeq(\bigcurlyvee_{i\geq j}\{\delta_{i}'.\gamma_{i}.\delta_j\}\,,\,\bigcurlyvee_{i\geq j}\{\gamma_j.\delta_{i}.\gamma_{i}'\}) \\
		&\simeq(\bigcurlyvee_{i\geq j}\{\delta_{i}'.f_{j,i}\}\,,\,\bigcurlyvee_{i\geq j}\{\gamma_j.f_{i,j}\}) \\
		&\simeq (\delta_j',\gamma_j').
	\end{align*}
	Thus, for each $j\in\omega$, $\mathcal{K}^{HPrj}_{K_j/}((\delta_j,\gamma_j)\,,\,(\delta_j',\gamma_j'))\neq\emptyset$, that is $\mathcal{K}^{HP}_{/K_j}(\gamma_j',\gamma_j)$ and $\mathcal{K}^{HE}_{K_j/}(\delta_j,\delta_j')$ spaces that are not empty. 
	
	\medskip Since $\mathcal{K}^{HP}_{/K_j}(\gamma_j',\gamma_j)\subseteq\mathcal{K}_{/K_j}(\gamma_j',\gamma_j)$ and by hypothesis $\gamma_j$ is an object final in $\mathcal{K}_{/K_j}$ for all $j\in\omega$, then $\mathcal{K}^{HP}_{/K_j}(\gamma_j',\gamma_j)$ and $\mathcal{K}^{HE}_{K_j/}(\delta_j,\delta_j')$ are contractible for each $j\in\omega$. Thus, the Kan complex $\mathcal{K}^{HPrj}_{K_j/}((\delta_j,\gamma_j)\,,\,(\delta_j',\gamma_j'))$ is contractible for each $j\in\omega$, that is, $(h,k)$ is unique (under homotopy) in the mediating morphism between $(K,\{(\delta_i,\gamma_i)\}_{i\in\omega})$ and $(K',\{(\delta_i',\gamma_i')\}_{i\in\omega})$.
\end{proof}

Therefore, the following corollary is an immediate consequence of the previous theorem.

\begin{corol}\label{Corollary-universal-colimit}
	The cocone $(K,\{(\delta_i,\gamma_i)\}_{i\in\omega})$ for the $\omega$-chain $(\{K_i\}_{i\in\omega},\{(f^+_i,f^-_i)\}_{i\in\omega})$ in $\mathcal{K}^{HPrj}$ is universal (a cocone colimit) iff $\Theta=\bigcurlyvee_{i\in\omega}\delta_i.\gamma_i\simeq I_K$.
\end{corol} 

\begin{defin}[Locally monotonic]
	Let $\mathcal{K}$ be a $0$-$\infty$-category. A functor $F:\mathcal{K}^{op}$$\times\mathcal{K}\rightarrow\mathcal{K}$	is locally h-monotonic if it is monotonic on the Kan complexes of 1-simplexes, i.e., for $f,f'\in \mathcal{K}^{op}(A,B)$ and $g,g'\in \mathcal{K}(C,D)$ one has 
	$$f\precsim f' \, , g\precsim g'\Longrightarrow F(f,g)\precsim F(f',g').$$ 
\end{defin}

\begin{prop}
	If $F:\mathcal{K}^{op}$$\times\mathcal{K}\rightarrow\mathcal{K}$ is locally h-monotonic and $(f^+,f^-)$, $(g^+,g^-)$ are h-projection pairs, then $F^{+-}((f^+,f^-),(g^+,g^-))$ is also an h-projection pair.
\end{prop}

\begin{defin}[Locally continuous]
	Let $\mathcal{K}$ be a $0$-$\infty$-category. A $F:\mathcal{K}^{op}\times\mathcal{K}\rightarrow\mathcal{K}$ is locally continuous if it is $\omega$-continuous on the Kan complexes of 1-simplexes. That is, for every directed diagram $\{f_i\}_{i\in\omega}$ in $\mathcal{K}^{op}(A,B)$, and every directed diagram $\{g_i\}_{i\in\omega}$ in $\mathcal{K}(C,D)$, one has
	\begin{center}
		$F(\bigcurlyvee_{i\in\omega}\{f_i\},\bigcurlyvee_{i\in\omega}\{g_i\})\simeq \bigcurlyvee_{i\in\omega}F(f_i,g_i).$ 
	\end{center} 
\end{defin}

\begin{rem}
	If $F$ is locally continuous, then it is also locally monotonic.
\end{rem}

\begin{teor}\label{continuous-functor}
	Let $\mathcal{K}$ be a $0$-$\infty$-category. Let also $F:\mathcal{K}^{op}\times\mathcal{K}\rightarrow\mathcal{K}$ be a locally continuous functor. Then the functor $F^{+-}:\mathcal{K}^{HPrj}\times\mathcal{K}^{HPrj}\rightarrow\mathcal{K}^{HPrj}$ is $\omega$-continuous. 
\end{teor}

\begin{rem}\label{Remark-Solution-HDE}
	Let $\mathcal{K}$ be a cartesian closed $0$-$\infty$-category, $\omega^{op}$-complete and with final object. Since the exponential functor  $\Rightarrow:\mathcal{K}^{op}\times\mathcal{K}\rightarrow\mathcal{K}$ and the diagonal functor $\Delta:\mathcal{K}\rightarrow\mathcal{K}\times\mathcal{K}$
	are locally continuous, by the Theorem \ref{continuous-functor}, the associated functors $$(\Rightarrow)^{+-}:\mathcal{K}^{HPrj}\times\mathcal{K}^{HPrj}\rightarrow\mathcal{K}^{HPrj}, \hspace{0.5cm} (\Delta)^{+-}:\mathcal{K}^{HPrj}\rightarrow\mathcal{K}^{HPrj}\times\mathcal{K}^{HPrj}$$
	are $\omega$-continuous. But composition of $\omega$-continuous functors is still an $\omega$-continuous functor. Thus, the functor 
	$$F=(\Rightarrow)^{+-}.(\Delta)^{+-}:\mathcal{K}^{HPrj}\rightarrow\mathcal{K}^{HPrj},$$
	is $\omega$-continuous. By Theorem \ref{point-fixed-theorem} the functor $F$ has a fixed point, that is, there is a vertex $K\in\mathcal{K}$ such that $K\simeq(K\Rightarrow K)$. The $\infty$-category of the fixed points of $F$ is denoted by $Fix(F)$.
\end{rem}

\section{Homotopy Domain Equation on $Kl(P)$}
\label{HDE on Kl(P)}

In this section we consider  $Kl(P)$ of \cite{MartinezHoDT} to be an $\infty$-category, in order to apply the homotopy domain theory of the previous section.

\bigskip For the next proposition, let 	$\mathscr{P}r^L_\kappa$ be the subcategory of $CAT_\infty$ whose objects are $\kappa$-compactly generated $\infty$-categories and whose morphisms are functors which preserve small colimits and $\kappa$-compact objects. Also, let $CAT_\infty^{Rex(\kappa)}$ denote the subcategory of $CAT_\infty$ whose objects are $\infty$-categories which admit $\kappa$-small colimits and whose morphisms are functors which preserve $\kappa$-small colimits. Finally, let $L^\ast:Cat_\infty\rightarrow Cat_\infty^{Rex(\kappa)}$ the functor which closes an $\infty$-category to an $\infty$-category which admits $\kappa$-small colimits, so $L^\ast A\simeq (PA)^\kappa$.  

\begin{prop}\label{limits-proposition-Kl(P)}
	The $\infty$-category $Kl(P)$ admits limits for $\omega^{op}$-diagrams in $\mathscr{P}r^L_\kappa$. 
\end{prop}
\begin{proof}[Proof]
	Given an $\omega^{op}$-diagram of pro-functors in $Kl(P)$ such that this is associated to $\omega^{op}$-diagram $p$ of functors 
	$$PK_0 \stackrel{f_0}{\longleftarrow} PK_1\stackrel{f_1}{\longleftarrow} PK_2\stackrel{f_2}{\longleftarrow}\cdots$$
	in the $\infty$-category $\mathscr{P}r^L_\kappa$, which admits all the small limits. Then, there is a limit $(T,\{\gamma_k\}_{k\in\omega})$ in $\mathscr{P}r^L_\kappa$ for this $\omega^{op}$-diagram $p$. Since $T$ is presentable, so for any regular cardinal $\kappa$, the full subcategory of $\kappa$-compacts $T^\kappa$ is essentially small \cite{DBLP:books/mk/Lurie}. Thus $(T^{\kappa},\{\gamma_k.i\}_{k\in\omega})$ is a cone for the $\omega^{op}$-diagram in $Kl(P)$, where $i$ is the inclusion functor $T\supseteq T^{\kappa}$. 
	
	\medskip Let $(T',\{\tau_k\}_{k\omega})$ be another cone for $(\{PK_k\}_{k\in\omega},\{f_k\}_{k\in\omega})$ in $Kl(P)$. Then, $(PT',\{\tau_k^\#\in\mathscr{P}r^L_\kappa(PT',PK_k) \}_{i\in\omega})$ is a cone from $\omega^{op}$-diagram $p$ in $\mathscr{P}r^L_\kappa$. Thus, there is a unique edge (under homotopy) $h:PT'\rightarrow T$ in $\mathscr{P}r^L_\kappa$ such that $\gamma_k.h\simeq\tau_k^\#$, for each $k\in\omega$. Applying the full faithful functor $(-)^\kappa:\mathscr{P}r^L_\kappa\rightarrow CAT_\infty^{Rex(\kappa)}$ \cite{DBLP:books/mk/Lurie}, one has that $h^\kappa\in Fun^\kappa (L^\ast T',T^\kappa)\simeq Fun(T',T^\kappa)$ is the unique edge (under homotopy) such that  $(\gamma_k.i).h^\kappa\simeq\tau_k^\#.j$, for each $k\in \omega$ with $j:L^\ast T'\rightarrow PT'$ being a Yoneda embedding. Thus, there is a unique (under homotopy) $h':T'\rightarrow T^\kappa$ such that $(\gamma_k.i).h'\simeq\tau_k$.
\end{proof}

Since each $Fun(A,PB)$ has initial object, any $Fun(A,PB)(F,G)$ is contractible or empty. Thus, $K(P)(A,B)\subseteq Fun(A,PB)$ admits a homotopy partial order (h.p.o.). For the following theorem, denote by $0_{PA,PB}$ in $Fun^L(PA,PB)$ as the constant functor in empty Kan complex $\emptyset$, that is
$$0_{PA,PB}f:=\boldsymbol{\lambda}x\in B.\emptyset$$

\begin{teor}
	The $\infty$-category $Kl(P)$ is a $0$-$\infty$-category.	
\end{teor}
\begin{proof}[Proof]
	\begin{enumerate}
		\item Since $PB$ is presentable, $Fun[A,PB]$ is presentable, thus $\langle Kl(P)(A,B),\precsim\rangle$ is complete. On the other hand, let $F$ be an object in  $Fun^L(PA,PB)$, then $$0_{PA,PB}fx=(\boldsymbol{\lambda}x\in B.\emptyset) x=\emptyset\subseteq Ffx,$$  for every object $f$ in $P A$ and $x$ in $B$. Thus, $0_{PA,PB}$ is an initial object in $Fun^L(PA,PB)$, i.e., $0_{A,B}$ is an initial object in $Kl(P)(A,B)$.
		\item Let $\{p_{i}\}_{i\in\omega}$ be a non-decreasing chain of morphisms in $Fun^L(PA,PB)$. By (1), the colimit  $\bigcurlyvee_{i\in\omega}p_{i}$ exists. 
		
		\begin{enumerate}
			
			\item[i.] First let's prove that for each object $z$ in $PA$, its colimit is given by 
			$$ (\bigcurlyvee_{i\in\omega}p_{i})z\simeq \bigcurlyvee_{i\in\omega}p_{i}z.$$
			Since $p_{i}\precsim\bigcurlyvee_{i\in\omega}p_{i}$, given a vertex $z$ of $PA$, by Definition \ref{h.p.o},  $p_{i}z\precsim(\bigcurlyvee_{i\in\omega}p_{i})z$. On the other hand, $\bigcurlyvee_{i\in\omega}p_{i}z$ is the supremum (under equivalence) of $ \{p_{i}z\} $, hence 
			$$p_{i}z\precsim\bigcurlyvee_{i\in\omega}p_{i}z\precsim(\bigcurlyvee_{i\in\omega}p_{i})z.$$
			Let $qz:=\bigcurlyvee_{i\in\omega}p_{i}z$, by Definition \ref{h.p.o}
			$$p_{i}\precsim q\precsim\bigcurlyvee_{i\in\omega}p_{i},$$
			but $\bigcurlyvee_{i\in\omega}p_{i} $ is the supremum (under equivalence) of $\{p_{i}\}_{i\in\omega}$, thus $ \bigcurlyvee_{i\in\omega}p_{i}\simeq q$, by Definition \ref{h.p.o},
			$$(\bigcurlyvee_{i\in\omega}p_{i})z\simeq\bigcurlyvee_{i\in\omega}p_{i}z.$$
			
			\item[ii.]Now let's prove that the composition is continuous on the right. Take a functor $F$ in $Fun_{\kappa}(PA',PA)$ and a vertex $z$ of $PA'$, then $Fz$ is a vertex of $PA$, by (i) we have
			\begin{align*}
				((\bigcurlyvee_{i\in\omega}p_{i}).F)z&\simeq(\bigcurlyvee_{i\in\omega}p_{i})(Fz) \\
				&\simeq\bigcurlyvee_{i\in\omega}p_{i}(Fz) \\
				&\simeq\bigcurlyvee_{i\in\omega}(p_{i}.F)z \\
				&\simeq(\bigcurlyvee_{i\in\omega}p_{i}.F)z,
			\end{align*}
			since $Kl(P)$ does have enough points, it follows 
			$Fun^L(PA',PB)$, 
			$$(\bigcurlyvee_{i\in\omega}p_{i}).F\simeq\bigcurlyvee_{i\in\omega}p_{i}.F.$$
			
			\item[iii.] Finally let's prove that the composition is continuous on the left. Let $G$ be a functor in $Fun^L(PB,PC)$ and $z$ an object in $PA$. By (i) and the continuity of $G$, we have
			\begin{align*}
				(G.\bigcurlyvee_{i\in\omega}p_{i})z&\simeq G((\bigcurlyvee_{i\in\omega}p_{i})z) \\
				&\simeq G(\bigcurlyvee_{i\in\omega}p_{i}z) \\
				&\simeq\bigcurlyvee_{i\in\omega}G(p_{i}z) \\
				&\simeq\bigcurlyvee_{i\in\omega}(G.p_{i})z \\
				&\simeq(\bigcurlyvee_{i\in\omega}G.p_{i})z,
			\end{align*}
			since $Kl(P)$ does have enough points, it follows
			$$G.\bigcurlyvee_{i\in\omega}p_{i}\simeq \bigcurlyvee_{i\in\omega}G.p_{i}.$$ 
		\end{enumerate}
		\item Let $F$ be an object in $Fun^L(PA,PB)$ and $f$ in $PA$, hence
		$$(0_{PB,PC}.F)f=0_{PB,PC}(Ff)= \boldsymbol{\lambda}x\in C.\emptyset=0_{PA,PC}f,$$ 
		that is, $0_{PB,PC}.F= 0_{PA,PC}$.
	\end{enumerate}
\end{proof}

\begin{prop}
	For any small $\infty$-category $A$, there is an h-projection from $A$ to $A\Rightarrow A$ in $Kl(P)$.
\end{prop}
\begin{proof}[Proof]
	
	We have that there is a diagonal functor $$\delta:PA\rightarrow [PA,PA]^L\simeq P(A^{op}\times A)= P(A\Rightarrow A),$$ where $[PA,PA]^L$ is the $\infty$-category of the functors which preserve  small colimits or left adjoints. Since $\delta$ preserves all  small colimits, by the Adjoint Functor Theorem, $\delta$  has a right adjoint $\gamma$ \cite{DBLP:books/mk/Lurie}. One the other hand, the diagonal functor $\delta$ is an h-embedding, then the unit is an equivalence, i.e., $\gamma.\delta\simeq I_{PA}$ and the counit is an h.p.o., that is, $\delta.\gamma\precsim I_{P(A\Rightarrow A)}$ in the c.h.p.o.\ $Kl(P)(A\Rightarrow A,A\Rightarrow A)$. Thus, $(\delta,\gamma)$ is a projection pair of $A$ to $A\Rightarrow A$ in $Kl(P)$, which we call the diagonal projection. 
\end{proof}
\begin{prop}
	There is a reflexive non-contractible Kan complex  in $Kl(P)$.
\end{prop}
\begin{proof}[Proof]
	First the trivial Kan complex $\Delta^0$ is a fixed point from endofunctor $FX=(X\Rightarrow X)$ on $Kl(P)$, since
	$$P\Delta^0\simeq P(\Delta^0\times \Delta^0)= P(\Delta^0\Rightarrow\Delta^0).$$
	that is, $\Delta^0\simeq (\Delta^0\Rightarrow\Delta^0)=F\Delta^0$ in $Kl(P)$. Let's suppose that all the Kan complexes in $Fix(F)$ are equivalent to $\Delta^0$. Since $Kl(P)$ contains all the small $\infty$-categories, then there is a small non-contractible Kan complex $K_0$, i.e.,
	$$\Delta^0\prec K_0\stackrel{(\delta_0,\gamma_0)}{\longrightarrow} FK_0$$ 
	in $(Kl(P))^{HPrj}$ for all $n\in\omega$, where $\delta_0$ is the diagonal functor, which has its equivalent functor in $\mathscr{P}^R_\kappa=(\mathscr{P}^L_\kappa)^{op}$ \cite{DBLP:books/mk/Lurie}. Since   $(\{F^iK_0\}_{i\in\omega},\{F^i(\gamma_0)\}_{i\in\omega})$ is an $\omega^{op}$-diagram in $\mathscr{P}^L_\kappa$, by Proposition \ref{limits-proposition-Kl(P)} and Theorem \ref{colimit-HPrj-category}, there is a colimit $(K,\{(\delta_{i,\omega},\gamma_{\omega,i})_{i\in\omega}\})$ in $(Kl(P))^{HPrj}$ for the $\omega$-diagram  $(\{F^iK_0\}_{i\in\omega},\{F^i(\delta_0,\gamma_0)\}_ {i\in\omega})$. Thus,
	$$\Delta^0\prec K_0\stackrel{\delta_{0,\omega}}{\hookrightarrow} K \in Fix(F)$$
	which is a contradiction. 
\end{proof}

\begin{defin}
	Let $X$ be a Kan complex. Define $\pi_0(X):=\pi_0(|X|)$ and $\pi_n(X,x):=\pi_n(|X|,x)$ for $n>0$, with $|\,\,\,|:Kan\rightarrow Top$ be the functor of geometric realization from the category of the Kan complexes to the category of the topological spaces.
\end{defin}

The fact that a Kan complex $X$ is not contractible does not imply that every vertex $x\in X$ contains relevant information, that is, the higher fundamental groups $\pi_n(|X|,x)$ are not trivial, nor that it contains holes in all the higher dimensions. To guarantee the existence of non-trivial Kan complexes as higher $\lambda$-models, we present the following definition.

\begin{defin}[Non-trivial Kan complex]
	A small Kan complex $X$ is non-trivial if
	\begin{enumerate}
		\item  $\pi_0(X)$ is infinite. 
		\item for each $n\geq 1$, there is a vertex $x\in X$ such that $\pi_n(X,x)\ncong\ast$.
	\item for each vertex $x$ of some non-degenerated horn in $X$, there is $n\geq 1$ such $\pi_n(X,x)\ncong\ast$.
	\end{enumerate}
	
\end{defin}

\begin{ejem}\label{Example1}
	For each $n\geq 0$, let  $S^n$ the sphere of Remark \ref{Remark}. Define the non-trivial Kan complex $B_0$ as the disjoint union:
	$$B_0=\coprod_{n<\omega}S^n.$$
\end{ejem}

 Furthermore, $\pi_{n}(S^n)\ncong\ast$ for all $n\geq 1$, and there is $k > n$ such that $\pi_{k}(S^n)\ncong\ast$ for each $n\geq 2$ \cite{DBLP:books/mk/Hatcher01}. 

\begin{ejem}\label{Example2}
	For each $n\geq 1$, let $D^{n+1}$ be a Kan complex,  such that its $k$-th face set the isomorphism  
	\begin{equation*}
		d_kD^{n+1}\cong
		\begin{cases}
			B^{n}& \text{if \, $k=0$,}\\
			S^n & \text{if \, $1\leq k\leq n+1$,}
		\end{cases}
	\end{equation*}
	Where $B^n$ is the union of the sphere $S^n$ with its interior. Define the non-trivial Kan complex $D_0$ as the disjoint union:
	$$D_0=\coprod_{n<\omega}D^{n+1}.$$
\end{ejem}

Note that each Kan complex $D^{n+1}$, with $n\geq 1$, corresponds to the sphere $S^{n}$ with $n+1$ holes. Besides that its higher groups have the same properties from Example \ref{Example1}, we also have the additional property $\pi_{n-1}(D^{n+1})\ncong\ast$ for all $n\geq 2$.

\begin{ejem}\label{Example3}
	For each $n\geq 2$, let $E^n$ be a Kan complex,  such that its $k$-th face set the isomorphism  
	\begin{equation*}
		d_kE^n\cong
		\begin{cases}
			\partial^{n-2}\Delta^{n-1}_\bullet& \text{if \, $0\leq k\leq 2$,}\\
			\partial^{n-k}\Delta^{n-1}_\bullet & \text{if \, $3\leq k\leq n$.}
		\end{cases}
	\end{equation*}
	Define the non-trivial Kan complex $E_0$ as the disjoint union
	$$E_0=\coprod_{n<\omega}E^n.$$ 
\end{ejem}

Where $\Delta^n_\bullet$ have the same vertices and faces of $\Delta^n$ but invertible 1-simplexes. Note that $E_0$ has more information than the non-trivial Kan complexes $B_0$ and $D_0$ from the previous examples, in the sense that for all $n\geq 2$, it satisfies the property $\pi_{k}(E^n)\ncong\ast$  for each $1\leq k\leq n-1$.

\begin{prop}
	For every non-trivial Kan complex $K_0$, there exists a non-trivial Kan complex $K\in Fix(F)$ above $K_0$.
\end{prop}
\begin{proof}[Proof.]
	Let $K_0$ be a non-trivial  Kan complex and $(K,\{(\delta_{i,\omega},\delta_{\omega,i})\}_{i\in\omega})$ the colimit from $\omega$-diagram $(\{F^iK_0\}_{i\in\omega},\{F^i(\delta_0,\gamma_0)\}_{i\in\omega})$ in $(Kl(P))^{HPrj}$, with $(\delta_0,\gamma_0)$ the first projection from $K_0$ to $K_1:=FK_0$. Let $z=(x,y)$ be a vertex of some $k$-simplex in $K_{i+1}$, with $k\geq 2$. By induction on $i$, there is $n_1,n_2\geq 1$, such that $\pi_{n_1}(K_i,x),\pi_{n_2}(K_i,x)\ncong\ast$. For any $n\in \{n_1,n_2\}$, one has 
	\begin{align*}
		\pi_n(K_{i+1},z)&=\pi_n(K_i\Rightarrow K_i,z) \\
		&=\pi_n(K_i\times K_i,(x,y)) \\
		&\cong\pi_n(K_i,x)\times\pi_n(K_i,y) \\
		&\ncong \ast
	\end{align*}
	Given any vertex $y$ of some $k$-simplex in $K$, with $k\geq 2$. There is $i\geq 0$ and $x\in K_i$ such that $\delta_{i,\omega}x\simeq y$. Since there is $\pi_n(K_i,x)\ncong\ast$, then 
	$$\pi_n(K,y)\cong\pi_n(K,\delta_{i,\omega}x)\ncong\ast.$$
\end{proof}

\begin{defin}[Non-Trivial Homotopy $\lambda$-Model]
	Let $\mathcal{K}$ be a Cartesian closed $\infty$-category with enough points. A Kan complex $K\in\mathcal{K}$ is a non-trivial homotopy $\lambda$-model if $K$ is a reflexive non-trivial Kan complex. 
\end{defin}

Note that every non-trivial homotopy $\lambda$-model is a homotopic $\lambda$-model as defined in \cite{Martinez} and \cite{MartinezHoDT}, which only captures information up to dimension 2 (for equivalences (2-paths) of 1-paths between points). While the non-trivial homotopy $\lambda$-models have no dimensional limit to capture relevant information, and therefore, those can generate a richer higher $\lambda$-calculus theory.

\begin{ejem}
	Given the non-trivial Kan complexes $B_0$, $D_0$ and $E_0$ of the Examples \ref{Example1}, \ref{Example2} and \ref{Example3} respectively. Starting from the diagonal projection as the initial projection pair,  these initial objects will generate the respective non-trivial Kan complexes $B$, $D$ and $E$ in $Fix(F)$. One can see that the non-trivial homotopy $\lambda$-model $E$  has more information than the homotopic $\lambda$-models $B$ and $D$.
\end{ejem}

\section{Conclusions and further work}
Some methods were established for solving homotopy domain equations, which further contributes to the project of a generalization of the Domain Theory to a Homotopy Domain Theory (HoDT). Using those methods of solving equations, it was possible to obtain some specific homotopy models in a Cartesian closed $\infty$-category, which could help to define a general higher $\lambda$-calculus theory. 

\bigskip Therefore, in this work, we would be facing the beginning of the semantics of another version of Martin-Löf's type theory (MLTT), but based on computational paths (\cite{Queiroz2016} and \cite{Ramos2017}). We will continue working on semantics for this version MLTT, where the types are interpreted as ordered Kan complexes, aiming at the development of new interactive theorem provers such as Coq, Lean, and similar systems.

\bigskip The above would be for the typed case for future work. For the case type-free, HoDT should continue to be developed, parallel to Dana Scott’s Domain Theory, and all the semantics of higher $\lambda$-calculus and its relation with HoTFT (Homotopy Type-Free Theory) \cite{Martinez2HoDTvsHoTT21}.

\bibliography{mybibfile}

\end{document}